\documentclass[letterpaper,10pt,conference]{IEEEtran}  % Comment this line out
\IEEEoverridecommandlockouts                              % This command is only
\usepackage{subfigure}
\usepackage{amsmath}
\usepackage{amsfonts}
 \usepackage[small,it]{caption2}

\usepackage{wrapfig}

\usepackage{graphics} % for pdf, bitmapped graphics files
\usepackage{epsfig} % for postscript graphics files
 \def\epsy{\varepsilon}

\def\liminf{\mathop{\rm lim{\,}inf}}
\def\argmin{\mathop{\rm arg{\,}min}}
\def\argmax{\mathop{\rm arg{\,}max}}

\def\sq{\hbox{\rlap{$\sqcap$}$\sqcup$}}
\def\qed{\ifmmode\sq\else{\unskip\nobreak\hfil
\penalty50\hskip1em\null\nobreak\hfil\sq
\parfillskip=0pt\finalhyphendemerits=0\endgraf}\fi\medskip}

\def\atop#1#2{\genfrac{}{}{0pt}{}{#1}{#2}}

 \def\FRAC#1#2#3{\genfrac{}{}{}{#1}{#2}{#3}}

\def\ddtp{{\mathchoice{\FRAC{1}{d^{\hbox to 2pt{\rm\tiny +\hss}}}{dt}}%
{\FRAC{1}{d^{\hbox to 2pt{\rm\tiny +\hss}}}{dt}}%
{\FRAC{3}{d^{\hbox to 2pt{\rm\tiny +\hss}}}{dt}}%
{\FRAC{3}{d^{\hbox to 2pt{\rm\tiny +\hss}}}{dt}}}}

\def\half{{\mathchoice{\FRAC{1}{1}{2}}%
{\FRAC{1}{1}{2}}%
{\FRAC{3}{1}{2}}%
{\FRAC{3}{1}{2}}}}

\def\eqdef{\mathbin{:=}}

\newtheorem{theorem}{Theorem}[section]

\newtheorem{proposition}[theorem]{Proposition}
\newtheorem{lemma}[theorem]{Lemma}

\def\Lemma#1{Lemma~\ref{t:#1}}
\def\Proposition#1{Proposition~\ref{t:#1}}
\def\Theorem#1{Theorem~\ref{t:#1}}

\def\Section#1{Section~\ref{#1}}

\def\Figure#1{Figure~\ref{f:#1}}

\def\eq#1/{(\ref{e:#1})}

\newcommand{\field}[1]{\mathbb{#1}}

\def\ind{\field{I}}

\def\Re{\field{R}}

\def\zstate{{\sf Z}}

\def\pz{\clP(\zstate)}

\def\transpose{{\hbox{\it\tiny T}}}

\newcounter{rmnum}
\newenvironment{romannum}{\begin{list}{{\upshape (\roman{rmnum})}}{\usecounter{rmnum}
\setlength{\leftmargin}{10pt} \setlength{\rightmargin}{8pt}
\setlength{\itemindent}{-1pt} }}{\end{list}}

\newcounter{anum}

\def\til={{\widetilde =}}

\def\clF{{\cal F}}

\def\clP{{\cal P}}

\newlength{\noteWidth}
 \long\def\notes#1{\ifinner
             {\tiny #1}
             \else
             \marginpar{\parbox[t]{\noteWidth}{\raggedright\tiny #1}}
             \fi}

\def\notes#1{}

\def\st{\hbox{\rm subject to\ }}

 \def\bfphi{\mbox{\protect\boldmath$\phi$}}

\def\Var{\hbox{\sf Var}\,}

\def\Rank{\hbox{\sf rank}\,}

\usepackage{paralist}
\long\def\notes#1{\ifinner
             {\footnotesize #1}
             \else
              \marginpar{\parbox[t]{\noteWidth}{\raggedright\footnotesize #1}}
               \fi}
\def\Var{\hbox{\sf Var}\,}
\def\DMM{D^{\hbox{\rm\tiny MM}}}

\def\EXPF{\mathcal{E}}
\def\VCdim{\mathit{VC}}
\def\setofhalfspace{\mathbf{H}}
\def\setB{B}

\def\Proj{\mathcal{P}}

\title{Feature Extraction for Universal Hypothesis Testing
\\
via Rank-Constrained Optimization}
\author{Dayu Huang and Sean Meyn\\
\authorblockA{Dept.\ of ECE and  CSL, UIUC, Urbana, IL 61801, U.S.A.\\dhuang8, meyn`{\tt at}'illinois.edu}
}

\usepackage{url}
\begin{document}
\urlstyle{same}
\maketitle

%\textbf{\large Needs to complete the simulation section and conclusion section. Needs some illustrative figures... I will finish those Before Jan 1st. }

%\textbf{\large Needs more reference... Is Merhav's paper on estimating order of the exponential family relevant? (Actually there is a lot of stuff to do following %that paper)}
%\notes{I'd say Merhav's paper is relevant.  I don't know if we have room for a reference}

%\notes{Big changes to the abstract - we have to say what it is we are after!}

\begin{abstract}
% This paper concerns the construction of universal tests for binary hypothesis testing, in which the alternate hypothesis is poorly modeled, and the observation space is large.  
% \notes{I changed state space to observation space - state space has a special meaning (e.g., in Markov models}
%   The \textit{mismatched universal test}
% is a feature-based technique for this purpose.   In prior work it is shown that its  finite-observation performance can be much better than the (optimal) Hoeffding test when the observations space is large.   However, good performance depends crucially on the choice of features in the mismatched test.  The contributions of this paper are summarized as follows:
% \begin{romannum}
% \item   We obtain upper and lower bounds on the number of \textit{easily distinguishable distributions} in an exponential family.  
% \item
% This motivates a new framework for feature extraction, cast as a rank-constrained optimization problem.  
% \item
% We obtain a gradient-based algorithm to solve the rank-constrained optimization problem.
% A proof of local convergence is provided.
% We demonstrate that this algorithm has a good performance in numerical experiments.
% \end{romannum}

This paper concerns the construction of tests for universal hypothesis testing problems, in which the alternate hypothesis is poorly modeled and the observation space is large. The mismatched universal test is a feature-based technique for this purpose. In prior work it is shown that its finite-observation performance can be much better than the (optimal) Hoeffding test, and good performance depends crucially on the choice of features. The contributions of this paper include:
\begin{romannum}
\item   We obtain bounds on the number of \emph{$\epsy$-distinguishable distributions} in an exponential family.  
\item
 This motivates a new framework for feature extraction, cast as a rank-constrained optimization problem.  
\item
We obtain a gradient-based algorithm to solve the rank-constrained optimization problem and prove its local convergence. \end{romannum}

\bigbreak


\noindent
% \textbf{AMS subject classifications:}~\parbox[t]{.6\hsize}{
%  Primary:
%  93E35, %Stochastic learning and adaptive control
%  49J15, %Optimal control problems involving ordinary differential equations
%  93C40 %Adaptive control
% \\
%  Secondary:
%  65C05, %Monte Carlo methods
%  93E20 %Optimal stochastic control
%  68M20 %Performance evaluation; queuing; scheduling
% }

\end{abstract}

\thispagestyle{empty}

%\newpage

\section{Introduction}
\label{intro}

\subsection{Universal Hypothesis Testing}

In universal hypothesis testing, the problem is to design a test to decide in favor of either of two hypothesis $H0$ and $H1$, under the assumption that we know the probability distribution $\pi^0$ under $H0$, but have uncertainties about the probability distribution $\pi^1$ under $H1$. One of the applications that motivates this paper is detecting abnormal behaviors \cite{den87p222}: In the applications envisioned, the amount of data from abnormal behavior is limited, while there is a relatively large amount of data for normal behavior. 

%The normal behavior is usually not completely characterized so the test should also be robust against uncertainties about the normal behavior. In this paper we mainly focus on the uncertainty about abnormal behavior.  \notes{This could be a distraction. No?}

To be more specific, we consider the hypothesis testing problem in which a sequence of observations $Z_1^n:=(Z_1,\dots,Z_n)$ from a finite observation space $\zstate$ is given, where $n$ is the number of samples. The sequence $Z_1^n$ is assumed to be i.i.d. with marginal distribution $\pi^i \in \pz$ under hypothesis $Hi$ ($i=0,1$), where $\pz$ is the probability simplex on $\zstate$.

Hoeffding~\cite{hoe65p369} introduced a universal test, defined using the empirical distributions and the Kullback-Leibler divergence. The empirical distributions $\{\Gamma^n: n \geq 1 \}$ are defined as elements of $\pz$ via,
\[
\Gamma^n(A) = \frac{1}{n} \sum_{k =1}^n \ind  \{Z_k \in A \}, \qquad A \subset \zstate.
\]
The Kullback-Leibler divergence for two probability distributions $\mu^1, \mu^0 \in \pz$ is defined as,
\begin{equation}
D(\mu^1 \| \mu^0) = \langle \mu^1 , \log(\mu^1/\mu^0) \rangle.\nonumber
%\label{e:Div}
\end{equation} 
where the notation $\langle \mu, f \rangle$ denotes expectation of $f$ under the distribution $\mu$, i.e., $\langle \mu, f\rangle=\sum_{z}\mu(z)f(z)$. The Hoeffding test is the binary sequence,
\begin{equation}
\phi^{\sf{H}}_n =\ind\{D(\Gamma^n\|\pi^0) \geq \eta\},\nonumber
%\label{e:uniH}
\end{equation}
where $\eta$ is a nonnegative constant. The test decides in favor of $H1$ when $\phi^{\sf{H}}=1$.

It was demonstrated in \cite{unnhuameysurvee09} that the performance of the Hoeffding test is characterized by both its error exponent and the variance of the test statistics. We summarize this in \Theorem{performhoe}. The error exponent is defined for a test sequence  $\bfphi\eqdef \{\phi_1,\phi_2,\dots\}$ adapted to $Z_1^n$ as
\begin{equation}
\begin{aligned}
J_{\phi}^0 &\eqdef \liminf_{n \to \infty} -\frac{1}{n}\log({\pi^0}\{
\phi_n=1\} ),\\
J_{\phi}^1 &\eqdef \liminf_{n \to \infty} -\frac{1}{n}\log({\pi^1}\{
\phi_n=0\} ).
%\label{e:ErrorProb}
\end{aligned}\nonumber
\end{equation}

\begin{theorem}\label{t:performhoe}
\begin{enumerate}
\item The Hoeffding test achieves the optimal error exponent $J_{\phi}^1$ among all tests satisfying a given constant bound $\eta \ge 0$ on the exponent $J_{\phi}^0$, i.e., $J_{\phi^{\sf{H}}}^0\geq \eta$ and 
\begin{equation}
J_{\phi^{\sf{H}}}^1=   \sup \{   J_{\phi}^1   :
            \ \mbox{\it subject to} \  J_{\phi}^0 \ge \eta\}, \nonumber
\end{equation} 
\item The asymptotic variance of the Hoeffding test depends on the size of the observation space. When $Z_1^n$ has marginal $\pi^0$, we have
\[
\lim_{n \to \infty} \Var [n D(\Gamma^n \| \pi^0)] =\half(|\zstate| -1).
\] 
\end{enumerate}
\end{theorem}
\Theorem{performhoe} is a summary of results from \cite{hoe65p369,unnhuameysurvee09}. The second result can be derived from \cite{wil38p60,clabar90p453,csishi04}. It has been demonstrated in \cite{unnhuameysurvee09} that the variance implies a drawback of the Hoeffding test, hidden in the analysis of the error exponent:  Although asymptotically optimal, this test is not effective when the size of the observation space is large compared to the number of observations. 

% the variance is approximately proportional to $(|\zstate|-1)/n^2$ under $\pi^0$
% 
% when the number of samples is finite, $\Gamma^n$ will approximately be distributed around $\pi^1$ according to a Gaussian distribution whose variance decreases as $n$ increases. Another way to understand the finite sample performance is to look at the bias and variance of the test statistics $D(\Gamma^n\|\pi^0)$. It was shown in \cite{unnhuameysurvee09} that the asymptotic bias and variance of $D(\Gamma^n\|\pi^0)$ also depends on the size of the observation space:  
% \notes{To be more precise, (not very!!)}
% \notes{DH: I am not sure I understand your comments. Please advise.}
% In particular, the variance is approximately proportional to $(|\zstate|-1)/n^2$ under $\pi^0$, 
% and has a similar dependence on $|\zstate|$ under $\pi^1$. This result can also be derived using results in \cite{clabar90p453}. These results demonstrate a drawback of the Hoeffding test that is hidden in the analysis of the error exponent:  Although asymptotically optimal, this test is not effective when the size of observation space is large comparing to the number of test samples. 

\subsection{Mismatched Universal Test}

It was demonstrated in \cite{unnhuameysurvee09} that the potentially large variance in the Hoeffding test can be addressed by using a generalization of the Hoeffding test called the \emph{mismatched universal test}, which is based on the relaxation of KL divergence introduced in \cite{abbmedmeyliz07p284}. The name of the mismatched divergence comes from literature on mismatched decoding \cite{merkaplapshi94p1953}. 
The mismatched universal test enjoys several advantages:
\begin{enumerate}
 \item It has smaller variance.
 \item It can be designed to be robust to errors in the knowledge of $\pi^0$.
 \item It allows us to incorporate into the test partial knowledge about $\pi^1$  (see \Lemma{pisenough}), 
 as well as other considerations such as the heterogeneous cost of incorrect decisions.
\end{enumerate}

The mismatched universal test is based on the following variational representation of KL divergence,
\begin{equation}
D(\mu\|\pi)=\sup_{f}\bigl(\langle\mu,f\rangle-\log(\langle\pi,e^f\rangle)\bigr)\label{e:klvar}
\end{equation}
where the optimization is taken over all functions $f\colon\zstate\to\Re$.
The supremum is achieved by the log-likelihood ratio. 

The mismatched divergence is defined by restricting the supremum in \eqref{e:klvar} to a function class $\clF$:
\begin{equation}
  \DMM_\clF(\mu\| \pi) := \sup_{f\in \clF}\bigl(\langle\mu,f\rangle-\log(\langle\pi,e^f\rangle)\bigr).
\label{e:Dmm}
\end{equation}
The associated mismatched universal test is defined as
\begin{equation}
\phi_n^{\sf{MM}} =\ind\{\DMM(\Gamma^n\|\pi^0) \geq \eta\}.\nonumber
%\label{e:uni}
\end{equation}
%\notes{Too distracting here:
%We remark that $\eta$ is better selected to vary with $n$ as suggested by analysis in \cite{unnhuameysurvee09}.
%}
%\begin{wrapfigure}{r}{1.70in}
%\vspace{-.4cm}
%%\begin{figure}[h]
% \Ebox{.95}{geometryLRTSVMentropyMM_Volos.eps}
% \vspace{-.15cm}
%\caption{Geometric representation of the mismatched universal test.}
%\label{f:entropyMM}
%%\end{figure}
%\vspace{-.2cm}
%\end{wrapfigure}

In this paper we restrict to the special case of a %finite-dimensional 
\emph{linear} function class:  $\clF=\bigl\{f_r \eqdef \sum_i^{d} r_i \psi_i  \bigr\}$
where
 $\{\psi_i\}$ is a set of \emph{basis} functions,  and $r$ ranges over $\Re^d$. We assume throughout the paper that $\{\psi_i\}$ is \emph{minimal}, i.e., $\{\mathbf{1},\psi_1,\ldots,\psi_d\}$ are linearly independent. The basis functions can be interpreted as \emph{features} for the universal test.
%We usually write $f_r:=\psi_r$ when the basis functions are clear from context. 
In this case, the definition \eqref{e:Dmm} reduces to the convex program,
\begin{equation}
\DMM(\mu\|\pi)=\sup_{r \in \Re^d}\bigl(\langle\mu,f_r\rangle-\log(\langle\pi,e^{f_r}\rangle).\nonumber
\end{equation}

The asymptotic variance of the mismatched universal test is proportional to the dimension of the function class $d$ instead of $|\zstate|-1$ as seen in the Hoeffding test:
\[
\lim_{n \to \infty} \Var [n \DMM(\Gamma^n \| \pi^0)] =\half d,
\] 
when $Z_1^n$ has marginal $\pi^0$ \cite{unnhuameysurvee09}. In this way we can expect substantial variance reduction by choosing a small $d$. The function class also determines how well the mismatched divergence $\DMM(\pi^1 \|\pi^0)$ approximates the KL divergence $D(\pi^1 \|\pi^0)$ for possible alternate distributions $\pi^1$and thus the error exponent of the mismatched universal test \cite{huaunnmeyveesur09p62}. In sum, the choice of the basis functions $\{\psi_i\}$ is critical for successful implementation of the mismatched universal test. The goal of this paper is to construct algorithms to construct a suitable basis.

\subsection{Contributions of this paper}
In this paper we propose a framework to design the function class $\clF$, which allows us to make the tradeoff between the error exponent and variance. One of the motivations comes from results presented in \Section{disdis} on the maximum number of \emph{$\epsy$-distinguishable distributions} in an exponential family, which suggests that it is possible to use approximately $d=\log(p)$ basis functions to design a test that is effective against $p$ different distributions. In \Section{s:featureextract} we cast the feature extraction problem as a rank constrained optimization problem, and propose a gradient-based algorithm with provable local convergence property to solve it. 

The construction of a basis studied in this paper is a particular case of the feature extraction problems that have been studied in many other contexts. In particular, the framework in this paper is connected to the exponential family PCA setting of \cite{coldassch01p617}. The most significant difference between this work and the exponential PCA is that our framework finds features that capture the \emph{difference} between distributions, and the latter finds features that are \emph{common} to the distributions considered. 

%\notes{This is a new paragraph}
The mismatched divergence using empirical distributions can be interpreted as an estimator of KL divergence. To improve upon the Hoeffding test, we may apply other estimators, such as those using data dependent features \cite{wankulver05p3064, qinkulver09p2392}, or those motivated by source-coding techniques \cite{zivmer93p1270} and others \cite{nguwaijor07}. Our approach is different from them in that we exploit the limited possibilities of alternate distributions. 
% Basis selection is similar to the feature extraction problems that have been studied in many other contexts. The framework in this paper is also similar to the exponential family PCA setting of \cite{coldassch01p617}. The most significant difference from these two references is that the former finds features that capture the \emph{difference} between distributions, and the latter finds features that are \emph{common} to these distributions. To the best of our knowledge, the feature extraction problem considered here has not been studied before. 

%\vspace{-0.02in}
\section{Distinguishable Distributions}
\label{disdis}

The quality of the approximation of KL divergence using the mismatched divergence depends on the dimension of the function class. The goal of this section is to quantify this statement. 
\subsection{Mismatched Divergence and Exponential Family}
We first describe a simple result suggesting how a basis might be chosen given a finite set of alternate distributions, so that the mismatched divergence is equal to the KL divergence for those distributions:
\begin{lemma}
\label{t:pisenough}
For any $p$ possible alternate distributions $\{\pi^1, \pi^2, \ldots, \pi^p\}$, absolutely continuous with respect to $\pi^0$, there exist $d=p$ basis functions $\{\psi_1, \ldots, \psi_d\}$ such that $\DMM(\pi^i\|\pi^0)=D(\pi^i\|\pi^0)$ for each $i$.
These functions can be chosen to be the log-likelihood ratios $\{\psi_i = \log(\pi^i/\pi^0)\}$.
 \qed
\end{lemma}
It is overly pessimistic to say that given $p$ distributions we require $d=p$ basis functions. In fact, \Lemma{perfectapprox} demonstrates that if all $p$ distributions are in the same $d$-dimensional \emph{exponential family}, then $d$ basis functions suffices. We first recall the definition of an exponential family: For a function class $\clF$ and a distribution $\nu$, the exponential family $\EXPF(\nu, \clF)$ is defined as: 
\[
\EXPF(\nu,\clF)=\{\mu: \mu(z)=\frac{\nu(z) e^{f(z)}}{\langle\nu, e^{f}\rangle}, f \in \clF\}.
\] 
We will restrict to the case of linear function class, and we say that the exponential family is $d$-dimensional if this is the dimension of the function class $\clF$. The following lemma is a reinterpretation of \Lemma{pisenough} for the exponential family: 
\begin{lemma}
\label{t:perfectapprox}
Consider any $p+1$ mutually absolutely continuous distributions $\{\pi^i: 0\leq i \leq p\}$. Then $\DMM_\clF(\pi^i\|\pi^j)=D(\pi^i\|\pi^j)$ for all $i\neq j$ if and only if $\pi^i \in \EXPF(\pi^0, \clF)$ for all $i$.
\end{lemma}

\subsection{Distinguishable Distributions}
Except in trivial cases, there are obviously infinitely many distributions in an exponential family. In order to characterize the difference between different exponential families of different dimension, we consider a subset of distributions which we call $\epsy$-distinguishable distributions.

The motivation comes from the fact that KL divergences between two distributions are infinite if neither is absolutely continuous with respect to the other, in which case we say they are \emph{distinguishable}. When the distributions are distinguishable, we can design a test that achieves infinite error exponent. For example, consider two distributions $\pi^0,\pi^1$ on $\zstate=\{z_1,z_2,z_3\}$: $\pi^0(z_1)=\pi^0(z_2)=0.5$; $\pi^1(z_2)=\pi^1(z_3)=0.5$. It is easy to see that the two error exponents of the test $\phi_n(Z_1^n)=\ind\{\Gamma^n(z_3)>0.2\}$ are both infinite. It is then natural to ask: Given $p$ distributions that are pairwise distinguishable, how many basis functions do we need to design a test that is effective for them? 

Distributions in an exponential family must have the same support. We thus consider distributions that are approximately distinguishable, which leads to the definitions listed below: Consider the set-valued function $F^\epsilon$ parametrized by $\epsilon>0$,
\[		F^\epsilon(x):=\{z: x(z) \geq \max_z(x(z))-\epsilon\}\]
\vspace{-0.2in}
\begin{romannum}
\item[$\bullet$]
Two distributions $\pi^1,\pi^2$ are \textit{$\epsilon$-distinguishable}  if $F(\pi^1)\setminus F(\pi^2) \neq \emptyset$ and $F(\pi^2)\setminus F(\pi^1) \neq \emptyset$.

\item[$\bullet$]
A distribution $\pi$ is called \textit{$\epsilon$-extremal} if $\pi(F^\epsilon(\pi)) \geq 1-\epsilon$,
and a set of distributions $\mathcal{A}$ is called $\epsilon$-extremal if every $\pi \in \mathcal{A}$ is $\epsilon$-extremal.

\item[$\bullet$]
For an exponential family $\EXPF$, 
the integer $N(\EXPF)$ is defined as the maximum $N$ such that there exists an $\epsilon_0>0$ such that for any $0<\epsy<\epsilon_0$, there exists an $\epsy $-extremal ${\mathcal{A}} \subseteq \EXPF$ such that $|\mathcal{A}| \geq N$ and any two distributions in $\mathcal{A}$ are $\epsy $-distinguishable. 
\end{romannum}

%For each $\clF$, there is many associated exponential family $\EXPF(\nu,\clF)$, but $N(\EXPF)$ does not depend on the distribution $\nu$: \begin{lemma}
%If $\nu_1$ and $\mu_2$ are absolutely continuous with respect to each other, then
%\[\EXPF(\nu_1,\clF)=\EXPF(\nu_2,\clF)\]
%\end{lemma}
One interpretation of the final definition is that the test using a function class $\clF$ 
is effective against $N(\EXPF)$ distributions, in the sense that the error exponents for the mismatched universal test are the same as for the Hoeffding test, where $\EXPF =\EXPF(\nu, \clF)$: 
\begin{lemma}\label{t:expfimplication}
Consider a function class $\clF$ and its associated exponential family $\EXPF =\EXPF(\nu, \clF)$, where $\nu$ has full support,  and define $N=N(\EXPF(\nu,\clF))$.   Then, there exists a sequence $\{A^{(1)}, A^{(2)}, \ldots, A^{(m)}: m\ge 1 \} $,  such that for each $k$ the set $A^{(k)}\subset \EXPF$ consists of $N$ distributions,
\[\DMM_\clF(\pi\|\pi')=D(\pi,\pi') \quad \textrm{for any $\pi,\pi' \in A^{(k)}$}\]
% $A^{(k)}=\{\pi^{0,(k)}, \pi^{1, (k)}, \ldots, \pi^{N-1, (k)}\}$, 
and
\begin{equation}
\lim_{k\to\infty}  \min_{\atop{\pi,\pi'\in A^{(k)} }{\pi\neq \pi'}}
\DMM_\clF(\pi\|\pi') = \infty. \nonumber
\end{equation}\vspace{-0.15in}\qed
\end{lemma}

%\notes{DH: Why change to $m$? 
%SM10:  Did I introduce a typo????
%\\I've changed `` associated exponential family'' to ``and its associated exponential family''
%\\
%With current assumption, we do not have minimality results. There are examples that we can have sets with more than $N$ distributions. The issue is that $\log(\mu\|\pi)$ can arbitrarily close to infinity even if they are not singular to each other. For example, consider the two Bernoulli distribution with parameters $p=10^{-5}$ and $p=10^{-300}$.}
Let $\mathcal{P}(d)$ denote the collection of all $d$-dimensional exponential families. Define $\bar{N}(d)=\max_{\EXPF \in \mathcal{P}(d)} N(\EXPF)$. In the next result we give lower and upper bounds on $\bar{N}(d)$, which imply that $\bar{N}(d)$ depends exponentially on $d$:

\begin{proposition}
%\label{t:lowerboundexp}
\label{t:lowerupperexp}
%\textbf{\large 
%Under some conditions, (iid?). DH: No condition. This Proposition has nothing to do with the observation sequence. Perhaps I missed your point?}
The maximum $\bar{N}(d)=\max_{\EXPF  } N(\EXPF)$ admits the following lower and upper bounds:
\begin{eqnarray}
\!\!\!
\!\!\!
\bar{N}(d) &\ge& \exp\big(\lfloor \frac{d}{2}\rfloor [\log(|\zstate|)-\log \lfloor \frac{d}{2}\rfloor -1]\big)
\label{e:lowerexp}
\\[.25cm]
\!\!\!
\!\!\!
\bar{N}(d) &\le&  \exp\big((d+1)(1+\log(|\zstate|)-\log(d+1))\big) 
\label{e:upperexp}
\end{eqnarray}
\end{proposition}

It is important to point out that $\bar{N}(d)$ is exponential in $d$. This answers the question asked at the beginning of this section: There exist $p$ approximately distinguishable distributions for which we can design an effective mismatched test using approximately $\log(p)$ basis functions.

%\textbf{How to motivate the particular definition? Perhaps stating that when the distributions are in the same exponential family, then the mismatched divergence and %divergence coincide.}

\section{Feature Extraction via \\Rank-constrained Optimization}\label{s:featureextract}

Suppose that it is known that the alternate distributions can take on $p$ possible values, denoted by $\pi^1,\pi^2,\ldots,\pi^p$.   Our goal is to choose
the function class $\clF$ of dimension $d$ so that the mismatched divergence approximates the KL divergence for these alternate distributions, while at the same time keeping the variance small in the associated universal test. The choice of $d$ gives the tradeoff between the quality of the approximation and the variance in the mismatched universal test. 
We assume that $0<D(\pi^i\|\pi^0)<\infty$ for all $i$.\footnote{In practice the possible alternate distributions will likely take on a continuum of possible values.    It is our wishful thinking that we can choose a finite approximation with $p$ distributions, and choose $d$ much smaller than $p$, and the resulting mismatched universal test will be effective against all alternate distributions. Validation of this optimism will be left to future work.}

We propose to use the solution to the following problem as the function class:
\begin{equation}\label{e:optproori}
\max_{\clF} \{\frac{1}{p}\sum_{i=1}^p \gamma^i \DMM_{\clF}(\pi^i\|\pi^0): \dim(\clF)\leq d\} 
\end{equation}
where $\dim{\clF}$ is the dimension of the function class $\clF$. The weights $\{\gamma_i\}$ can be chosen to reflect the importance of different alternate distributions. This can be rewritten as the following rank-constrained optimization problem: 
\begin{equation}\label{e:optpro}
\begin{array}{rl}
\max & \frac{1}{p}\sum_{i=1}^p \gamma^i \bigl(\langle\pi^i,X_i\rangle-\log(\langle\pi^0,e^{X_i}\rangle\bigr) 
\\[.25cm]
\st & \Rank(X)\leq d
\end{array}
\end{equation}
where the optimization variable $X$ is a $p\times |\zstate|$ matrix,   
and $X_i$ is the $i$th row of $X$, interpreted as a function on $\zstate$. 
Given an optimizer $X^*$, we choose $\{\psi_i\}$ to be the set of right singular vectors of $X^*$ corresponding to nonzero singular values. %$\psi_i(k) = X^*_{ik}$,  $1\le k\le |\zstate|$,  $1\le i\le d$.
%
%The choice of $d$ gives a trade-off between the variance of the test and the quality of the approximation: when $d$ is smaller, the variance is smaller; when $d$ is larger, the approximation is better. The analysis in \Section{disdis} suggests that an effective test may be obtained even when $d$ is much smaller than $p$. 

\vspace{-0.03in}
\subsection{Algorithm}\label{subsecalgorithm}
The optimization problem in \eqref{e:optpro} is not a convex problem since it has a rank constraint. It is generally very difficult to design an algorithm that is guaranteed to find a global maximum. The algorithm proposed in this paper is a generalization of the Singular Value Projection (SVP) algorithm of \cite{mekjaidhi09} designed to solve a low-rank matrix completion problem. It is globally convergent under certain conditions valid for matrix completion problems.  However, in this prior work the objective function is quadratic;  we are not aware of any prior work generalizing these algorithms to the case of a general convex objective function. 

Let $h(X)$ denote the objective function of \eqref{e:optpro}. Let $\mathcal{S}$ denote the set of matrices satisfying $\Rank(X)\leq d$. Let $\Proj_{\mathcal{S}}$ denote the projection onto $\mathcal{S}$: 
\[\Proj_{\mathcal{S}}(Y)=\argmin\{\|Y-X\|: \Rank(X)\leq d\}.\]
where we use $\|\cdot\|$ to denote the Frobenius norm.
The algorithm proposed here is defined as the following iterative gradient projection: 
\begin{enumerate}
 \item $Y^{k+1}=X^k+\alpha^k \nabla h(X^k)$.
 \item $X^{k+1}=\Proj_{\mathcal{S}}(Y^{k+1})$.
\end{enumerate}
The projection step is solved by keeping only the $d$ largest singular values  of $Y^{k+1}$.
The iteration is initialized with some arbitrary $X^0$ and is stopped when the $\|X^{k+1}-X^{k}\| \leq \epsilon$ for some small $\epsilon>0$.
%$h(X^k)$ stops increasing for several iterations.
% \notes{Not sure about this stopping criterion} \notes{We should talk about it. I agree that we might have better choice. This is what I am using now//....}
\vspace{-0.05in}
\subsection{Convergence Result}
We can establish local convergence:
\begin{proposition}\label{t:localconverge}
Suppose $\bar{X}$ satisfies $\Rank(\bar{X})=d$ and is a local maximum, i.e. there exists $\delta>0$ such that for any matrix $X \in \mathcal{S}$ satisfying $\|X-\bar{X}\| \leq \delta$, we have $h(\bar{X}) > h(X)$. Choose $\alpha^k=\alpha$ for all $k$ where $0<\alpha<2/(\frac{1}{p}\max_i \gamma^i)$. Then there exists a $\delta'>0$ such that if $X^0$ satisfies $\|X^0-\bar{X}\| \leq \delta'$ and $\Rank(X^0)\leq d$, then $X^k \rightarrow \bar{X}$ as $k \rightarrow \infty$. Moreover, the convergence is geometric. \qed
\end{proposition}
%\notes{SM: We must stress that all that is required is local strict convexity}\notes{DH: I do not understand. We are addressing a particular objective function for which the local strict convexity and Lipschitz gradient always holds.}

Let $\mathcal{H}$ denote the hyperplane $\mathcal{H}=\{\bar{X}W_1+W_2\bar{X}: W_1 \in \Re^{n \times n}, W_2 \in \Re^{p \times p}\}$. The main idea of the proof is that near $\bar{X}$ the set $\mathcal{S}$
can be approximated by this hyperplane $\mathcal{H}$, as demonstrated in \Lemma{approx1}.% and \Lemma{approx2}. 
\begin{lemma}\label{t:approx1}
%Consider any $\bar{X}$ satisfying $\Rank(\bar{X})=d$. 
There exist $\delta>0$ and $M>0$ such that: 1) for any $X \in \mathcal{S}$ satisfying $\|X-\bar{X}\| \leq \delta$, there exists $Z \in \mathcal{H}$ such that $\|Z-X\|\leq M\|X-\bar{X}\|^2$; 2) for any $Z \in \mathcal{H}$ satisfying $\|Z-\bar{X}\| \leq \delta$, 
there exists $X \in \mathcal{S}$ satisfying $\|X-Z\|\leq M\|Z-\bar{X}\|^2$.
\end{lemma}
% \begin{lemma}\label{t:approx2}
% %Consider any $\bar{X}$ satisfying $\Rank(\bar{X})=d$. 
% There exist $\delta>0$ and $M>0$ such that for any $Y \in \mathcal{H}$ satisfying $\|Y-\bar{X}\| \leq \delta$, 
% there exists $X \in \mathcal{S}$ satisfying $\|X-Y\|\leq M\|Y-\bar{X}\|^2$.
% \end{lemma}

%We remark that $\delta$ and $M$ in \Lemma{approx1} and \Lemma{approx2} can depend on $\bar{X}$.
%\Lemma{approx1} is proved using the matrix perturbation theory. \fixme{Add reference} \Lemma{approx2} is proved using a construction of a low-rank matrix. 
%The proofs of these and the following lemmas are left to the full version of this paper.

Let $Z^{k}=\Proj_{\mathcal{H}}(Y^{k})$, i.e., the projection of $Y^k$ onto $\mathcal{H}$. We obtain from \Lemma{approx1} that $Z^k$ is close to $X^k$ as follows:
\begin{lemma}\label{t:projerror}
Consider any $\bar{X}$ satisfying $\Rank(\bar{X})=d$. There exist $\delta>0$ and $M>0$ such that if $\|Z^k-\bar{X}\|\leq \delta$, then $\|Z^k-X^k\| \leq M\|{Y}^k-\bar{X}\|^\frac{3}{2}$.
\end{lemma}

% \begin{proof}[Proof of \Lemma{projerror}]
% Let $\hat{Y}^k=\argmin \{\|X-\check{Y}^k\|: \Rank(X) \leq d\}$. Let $\check{X}^k=\mathcal{P}_{\mathcal{H}}(X^k)$. 
% Since $\check{X}^k$ and $\check{Y}^k$ both belong to $\mathcal{H}$, we obtain
% \[\|Y^k-\check{X}^k\|\geq \|Y^k-\check{Y}^k\|.\]
% Since $X^k$ is the projection of $Y^k$ onto $\mathcal{S}$ and $\hat{Y}^k$ also belongs to this set, we obtain
% \[\|Y^k-\hat{Y}^k\|\geq \|Y^k-X^k\|.\]
% Using triangular inequality, we obtain
% \[\|Y^k-\check{X}^k\|\leq \|\check{Y}^k-X^k\|+\|\check{Y}^k-Y^k\| + \|\check{X}^k-X^k\|\]
% Since $\bar{X}$ and $\check{Y}^k$ both belongs to the hyperplane $\mathcal{H}$, we have 
% \[\|\bar{X}-\check{Y}^k\| \leq \|\bar{X}-Y^k\|.\]
% Since $\check{X}^k$ is the projection of $Y^k$ onto $\mathcal{S}$ , we obtain 
% \[\|\bar{X}-\check{X}^k\|\leq \|\bar{X}-Y^k\|+ \|\check{X}^k-Y^k\| \leq 2\|\bar{X}-Y^k\|.\]
% Applying \Lemma{approx1} and \Lemma{approx2}, we obtain 
% \[\|Y^k-\check{X}^k\|\leq \|\check{Y}^k-X^k\|+ M_1\|\bar{X}-Y^k\|^2 + 4M_2\|\bar{X}-Y^k\|^2\]
% 
% Since $\mathcal{H}$ is a hyperplane, we have
% \[\|\check{X}^k-\check{Y}^k\|^2=\|Y^k-\check{X}^k\|^2-\|Y^k-\check{Y}^k\|^2\]
% Therefore, 
% \begin{eqnarray}
% \|\check{X}^k-\check{Y}^k\|&=&(\|Y^k-\check{X}^k\|-\|Y^k-\check{Y}^k\|)^\half\nonumber\\
% &&\times (\|Y^k-\check{X}^k\|+\|Y^k-\check{Y}^k\|)^\half\nonumber\\
%  &\leq&(M_1+4M_2)^\half\|\bar{X}-Y^k\| \sqrt{2\|\bar{X}-Y^k\|}\nonumber\\
%  &=&M\|\bar{X}-Y^k\|^{\frac{3}{2}}
% \end{eqnarray}
% 
% \end{proof}

\begin{lemma}
Gradients of $h(X)$ are Lipschitz with constant $L=\frac{1}{p}\max_i \gamma^i$, i.e. $\|\nabla h(X_1)-\nabla h(X_2)\| \leq L\|X_1-X_2\|$.
\end{lemma}

\begin{lemma}
Suppose $\bar{X}$ is a local maximum in $\mathcal{S}$ and $\Rank(\bar{X})=d$. Then $\bar{X}$ is also a local maximum in $\mathcal{H}$.
\end{lemma}

\begin{proof}[Outline of Proof of \Proposition{localconverge}]
Using standard results form optimization theory, we can prove that for any small enough $\delta>0$, if $\|X^k-\bar{X}\| \leq \delta$ and $\alpha < \frac{2}{L}$, then $\|Z^{k+1}-\bar{X}\|\leq q \|X^k-\bar{X}\|$ for some $q<1$ where $q$ could depend on $\delta$, and $\|Y^{k+1}-\bar{X}\|\leq \|X^k-\bar{X}\|$. Thus, we can choose a $\delta$ small enough so that
$M\delta^\frac{1}{2}\leq \frac{1-q}{2}$. With this choice, we have
\begin{eqnarray}
\|X^{k+1}-\bar{X}\| &\leq& \|Z^{k+1}-\bar{X}\|+\|Z^{k+1}-X^{k+1}\|\nonumber\\ 
&\leq& \|Z^{k+1}-\bar{X}\|+M\delta^{\half}\|{Y}^{k+1}-\bar{X}\|\nonumber\\ 
&\leq& (q+\half(1-q)) \|X^k-\bar{X}\|.\nonumber
\end{eqnarray}
\Proposition{localconverge} then follows from induction. 
\end{proof}

\section{Simulations}
%\vspace{-0.02in}
We consider probability distributions in an exponential family of the form $\pi^i(z)=\exp\{\sum_{k=1}^{q} \theta_{i,k} \psi_i(z)+\sum_{i=k}^{q'}\theta'_{i,k} \psi'_i(z)\}$. We first randomly generate $\{\psi_i\}$ and $\{\psi'_i\}$ to fix the model. A distribution is obtained by randomly generating $\{\theta_{i,k}\}$ and $\{\theta'_{i,k}\}$ according to uniform distributions on $[-1,1]$ and $[-0.1,0.1]$, respectively. In application of the algorithm presented in \Section{subsecalgorithm}, the bases $\{\psi_i\}$ and $\{\psi'_i\}$ are not given. This model can be interpreted as a perturbation to $q$-dimensional exponential family with basis $\{\psi_i\}$.  

In the experiment we have two phases: In the feature extraction (training) phase, we randomly generate $p+1$ distributions, taken as $\pi^0, \ldots, \pi^p$. We then use our techniques in \eqref {e:optproori} with the proposed algorithm to find the function class $\clF$. The weights $\gamma_i$ are chosen as $\gamma^i=1/D(\pi^i\|\pi^0)$ so 
that the objective value is no larger than $1$. In the testing phase, we randomly generate $t$ distributions, denoted by $\mu^1,\ldots, \mu^t$. We then compute the average of $\DMM(\mu^i\|\pi^0)/D(\mu^i\|\pi^0)$. 

For the experimental results shown in \Figure{dmmtraining}, the parameters are chosen as $q=8$, $q'=5$, and $t=500$. Shown in the figure is an average of $\DMM(\pi^i\|\pi^0)/D(\pi^i\|\pi^0)$ (for training) as well as $\DMM(\mu^i\|\pi^0)/D(\mu^i\|\pi^0)$ (for testing) for two cases: $p=50$ and $p=500$.
We observe the following:
\begin{compactenum}
\item The objective value increases gracefully as $d$ increases. For $d\geq7$, the values are close to 1.
\item The curve for training and testing are closer when $p$ is larger, which is expected. 

%\footnote{The absolute difference between the value of $\DMM/D$ for training and testing might not be a good criterion since both are $1$ when $d=n-1$. }
\end{compactenum}
\begin{figure}%{r}{2.90in}
\vspace{-.4cm}
\centering
\includegraphics[width=0.4\textwidth]{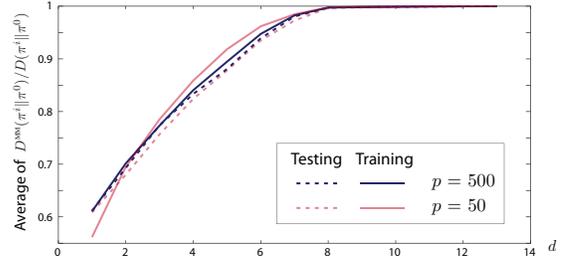}
\vspace{-.1cm}
\caption{Dashed curve: average of $\DMM(\mu^i\|\pi^0)/D(\mu^i\|\pi^0)$. Solid curve: average of $\DMM(\pi^i\|\pi^0)/D(\pi^i\|\pi^0)$}
\label{f:dmmtraining}
\vspace{-.5cm}
\end{figure}
\section{Conclusions}
The main contribution of this paper is a framework to address the feature extraction problem for universal hypothesis testing, cast as a rank-constrained optimization problem. This is motivated by results on the number of easily distinguishable distributions, which demonstrates that it is possible to use a small number of features to design effective universal tests for a large number of possible distributions. We propose a gradient-based algorithm to solve the rank-constrained optimization problem, and the algorithm is proved to converge locally. Directions considered in current research include: applying the nuclear-norm heuristic \cite{fazhinboy01p4734} to solve the optimization problem \eqref{e:optproori}, applying this framework to real-world data, and extension of this framework to incorporate other form of partial information. 

\appendix
 
\subsection{Proof of the lower bound in \Proposition{lowerupperexp}}
We give a constructive proof of the lower bound \eqref{e:lowerexp} by combining ideas in \Lemma{nonsingularcons} and \ref{t:singularconstruction}.

\begin{lemma}\label{t:nonsingularcons}
$\bar{N}(2)\geq |\zstate|$.
\end{lemma}
\begin{proof}
%We can assume without loss of generality that $\zstate=\{1, 2, \ldots, |\zstate|\}$.
We pick the following two basis functions $\psi_1,\psi_2$:
\begin{equation}
\begin{aligned}
 \psi_1&=[|\zstate|-1,|\zstate|-2, \ldots, 0], \\ \quad \textrm{and  } \psi_2&=[1, 1.5, \sum_{j=0}^{2} 2^{-j}, \ldots, \sum_{j=0}^{|\zstate|-1} 2^{-j}].
\end{aligned}\label{e:nonsingularfunction}
\end{equation}
%\begin{equation}\label{e:nonsingularfunction}
% \psi_1=[|\zstate|-1,\ldots, 0],  \psi_2=[1, \sum_{j=0}^{1} 2^{-j}, \ldots, \sum_{j=0}^{|\zstate|-1} 2^{-j}].
%\end{equation}

%In another word, $\psi_2(i)=\sum_{j=0}^{i-1} 2^{-j}$.  
%Let $\EXPF$ denote the exponential family associated to $\psi_1$ and $\psi_2$ with $\nu$ being the uniform distribution over $\zstate$.
\vspace{-0.08in}
For $1 \leq k \leq |\zstate|$, define $u^k$ as $u^k=\psi_1+2^{k-0.5}\psi_2$. Assuming without loss of generality that $\zstate=\{1, \ldots, |\zstate|\}$, we have $\argmax_z u^k(z)=k$ .
%\end{equation} 
% It is easy to see that 
% \[u^k(i)=|\zstate|-i+ \sum_{j=0}^{i-1} 2^{-j}2^{k+0.5}=|\zstate|-i+\sum_{j=0}^{i-1} \frac{1}{\sqrt{2}}2^{k-j}.\] Thus, 
% \[u^k(i+1)-u^k(i)=-1+\frac{1}{\sqrt{2}}2^{k-i}\] For $i\geq k$, we have $u^k(i+1)< u^k(i)$. For $i <k$, we have $u^k(i+1) > u^k(i)$. Thus, $u^k(k) > u^k(i)$ for $i \neq k$. 

Now, for any $\beta >0$, $1 \leq k \leq |\zstate|$, define the distribution
\[
\pi^{k,\beta}(z)=C\exp\{\beta u^k(z)\}.
\] 
where $C$ is a normalizing constant. 
%It is easy to see from \eqref{e:maxu} that for any $k ,k'$ such that $1 \leq k ,k' \leq n$, 
%\[\lim_{\beta \rightarrow \infty} D(\pi^{k,\beta}\|\pi^{k', \beta})=\infty.\] There are only finite number of choices of $k, k'$, thus the convergence is uniform. 
%Thus, for any $M>0$, there exists $\beta_M$ such that for any $\beta \geq \beta_M$, for any $k ,k'$ such that $1 \leq k ,k' \leq n$, $D(\pi^{k,\beta}\|\pi^{k', \beta})\geq M$. In sum, we have proved that for any $M$, there is a subset of $\EXPF$ that has cardinality $n$ and is $M$-separated. 
Since there are only finite choices of $k$, for any small enough $\epsilon$, there exists $\beta_0$ such that for $\beta \geq \beta_0$, $\{\pi^{k,\beta}, 1\leq k \leq |\zstate|\}$ are $\epsilon$-extremal and any two distributions in $\{\pi^{k,\beta}, 1\leq k \leq |\zstate|\}$ are $\epsilon$-distinguishable. 
\end{proof}

\begin{lemma}\label{t:singularconstruction}
$\bar{N}(d)\geq {d \choose \lfloor d/2 \rfloor}$
\end{lemma}
\begin{proof}
Take $\psi_k(z)=\ind\{z=k\}$ for $1\leq k \leq d$. 
\end{proof}

\begin{proof}[Outline of proof of the lower bound]
% The basis function used in the construction is the Kronecker product of the indicator function used for \Lemma{singularconstruction} of length $\lfloor d/2\rfloor$ and the basis functions used for \Lemma{nonsingularcons} of length $\lfloor \frac{|\zstate|}{\lfloor d/2 \rfloor}\rfloor$:
The basis functions used in the construction are the Kronecker products of basis functions used for \Lemma{singularconstruction} and \Lemma{nonsingularcons}. 

Let $J=\lfloor  {|\zstate|}/{\lfloor \half d \rfloor}\rfloor$. Let $\bar{\psi}_1, \bar{\psi_2}$ denote the basis function defined in \eqref{e:nonsingularfunction} with $|\zstate|$ replaced by $J$.
The basis functions used for the lower bound are given by 
\begin{equation}
 \begin{aligned}
  \psi_{k}(i+jJ)&=\ind\{j=k-1\}\bar{\psi}_1(i),\quad \textrm{for $1\leq k \leq \lfloor \half d\rfloor$},\\
  \psi_{k+\lfloor d/2\rfloor}(i+jJ)&=\ind\{j=k-1\}\bar{\psi}_2(i),\quad \textrm{for $1\leq k \leq \lfloor \half d\rfloor$}.
 \end{aligned}\nonumber
\end{equation}
\end{proof}
\vspace{-0.2in}
\subsection{Proof of the upper bound in \Proposition{lowerupperexp}}
The main idea of the proof of \eqref{e:upperexp}  is to relate this bound  to VC dimension.  We first obtain an elementary upper bound. 
\begin{lemma}\label{t:NpandhatNp}
$N(\EXPF)\leq \hat{N}(\EXPF)$, where
\[
 \hat{N}(\EXPF)=|\{F^\epsilon(\sum_l r_l \psi_l): r \in \Re^{d}, \epsilon>0\}|.
\]
\end{lemma}
\vspace{-0.1in}
\begin{proof}
By definition if a subset $A$ of $\EXPF$ is $\epsilon$-extremal, and any two distributions in $A$ are $\epsilon$-distinguishable, then for any two distributions $\pi^i, \pi^j \in A$, there exists $\epsilon_1, \epsilon_2>0$ such that $F^{\epsilon_1}(\log(\pi^1))\neq F^{\epsilon_2}(\log(\pi^2))$. 
\end{proof}

Let $\setofhalfspace$ denote the set of all the half space in $\Re^d$, and let $\VCdim(\setofhalfspace)$ denote the VC dimension of $\setofhalfspace$.   It is known that  $\VCdim(\setofhalfspace)=d+1$  \cite[Corollary of Theorem 1]{bur98p121}.

For any finite subset $\setB$ of $\Re^d$, define $\tau(\setB)=|\{h \cap \setB: h \in \setofhalfspace\}|$. In other words, $\tau(\setB)$ is the number of subsets one can obtain by intersecting $\setB$ with half-spaces from $\setofhalfspace$. A bound on $\tau(\setB)$ is given by Sauer's lemma:
\begin{lemma}[Sauer's Lemma]
\label{t:Sauer}
The following bound holds whenever $|\setB| \geq \VCdim(\setofhalfspace)$: 
\[
\tau(\setB)  \leq (\frac{e |\setB|}{\VCdim(\setofhalfspace)})^{\VCdim(\setofhalfspace)}.
\]
\end{lemma}

Consider any $d$-dimensional exponential family $\EXPF$ with basis $\{\psi_l, 1\leq l \leq d\}$. 
Define a set of function $\{y^i \}\subset  \Re^d$ via,  
\[
y^i_j=\psi_j(i),\qquad 1\leq i \leq |Z|,\  1\le j\le d. 
\] 
In other words, if we stack $\{\psi_l\}$ into a matrix so that each $\psi_l$ is a row, then $\{y^i\}$ are the columns of this matrix. Let $B(\EXPF)=\{y^i, 1\leq i \leq |Z|\}$. The following lemma connects $\tau(B(\EXPF))$ to $\hat{N}(\EXPF)$. 

\begin{lemma}\label{t:hatNandhalf-space}$\hat{N}(\EXPF) \leq \tau(B(\EXPF)).$
\end{lemma}
\begin{proof}
For given $r \in \Re^d$ and $\epsilon>0$, denote $I=F^\epsilon(\sum_l r_l \psi_l)$. 
By the definition of $F^\epsilon$ we have $I=\{i: r^\transpose y^i \geq \sup_z(\sum_l r_l \psi_l(z))-\epsilon\}$. Therefore, there exists $b$ such that $r^\transpose y^i \geq b$ for all $i \in I$, 
and $r^\transpose y^i < b$ for all $i \notin I$. That is, $I$ is the subset of $\{y^i\}$ that lies in the half space $\{y: r^\transpose y\geq b\}$. Thus, $\{y^i: i \in I\} \in \{h \cap B(\EXPF): h \in H\}$. Since this holds for any element in $\{F^\epsilon(\sum_l r_l \psi_l): r \in \Re^{d}, \epsilon>0\}$, we obtain the result.
\end{proof}

% \begin{proof}\textbf{The proof will be revised later. Please skip it for now!}
% By the definition of $\hat{N}(\EXPF)$, we can pick a set $\Theta$ of $\hat{N}(\EXPF)$ vectors from $\Re^d$ so that $|F(\Theta)|=\hat{N}(\EXPF)$. Consider any $r \in \Theta$, let $I=F(\sum_l r_l \psi_l)$. By the definition of $F$, we have that $I=\{i: r^\transpose y^i=\sup_z(\sum_l r_l \psi_l(z))\}$. Therefore, there exists $b$ such that $r^\transpose y^i \geq b$ for all $i \in I$ and $r^\transpose y^i < b$ for all $i \notin I$. In other words, $I$ is the subset of $\{y^i\}$ that lies in the half space defined by $\{y: r^\transpose y\geq b\}$. Thus, $F(\Theta) \subseteq \{h \cap B(\EXPF): h \in H\}$. 
% \end{proof}

%Based on the foregoing we obtain the desired bound:

\begin{proof}[Proof of the upper bound]
We obtain \eqref{e:upperexp} on combining \Lemma{NpandhatNp}, \Lemma{Sauer} and \Lemma{hatNandhalf-space}, 
together with the identity $\VCdim(\setofhalfspace)=d+1$.
\end{proof}

% We investigate the number of easily distinguishable distributions in an exponential family and demonstrate that it is possible to use a small number of features to design effective tests for a large number of possible distributions. We propose a framework using a rank-constrained optimization problem and propose a gradient-based algorithm to solve it. We demonstrate the effectiveness of this algorithm via analysis and numerical experiments. Future work includes: 1) Investigate the number of easily distinguishable distributions in a randomly chosen exponential family. 2) Analyze the global convergence property of the proposed algorithm. 3) Extend the framework to incorporate other form of partial information on the hypotheses.

\vspace{0.05in}
\paragraph*{Acknowledgment} 
This research was partially supported by AFOSR under grant AFOSR FA9550-09-1-0190 and NSF under grants NSF CCF 07-29031 and
 NSF CCF 08-30776.   Any opinions, findings, and conclusions or recommendations expressed in this material are those of the authors and do not
necessarily reflect the views of the AFOSR or NSF.

%\vspace{-0.04in}
\bibliographystyle{IEEEtran}
\bibliography{IEEEabrv,DH_String,Dayu_bibtex}
%\section*{Acknowledgment}
%Research supported in part by  National Science Foundation contracts  *****\fixme{Add Funding Information}
\end{document}